\documentclass[reqno]{article}

\usepackage{amsthm}
\usepackage{amsmath}
\usepackage{amssymb}
\usepackage{graphicx}
\usepackage{color}
\usepackage{bbm}
\usepackage[top=2.5cm, bottom=2.5cm, left=2.5cm, right=2.5cm]{geometry}
\usepackage{enumerate}
\usepackage{faktor}
\usepackage{nicefrac}
\usepackage[nottoc]{tocbibind}
\usepackage{slashed}
\usepackage{authblk}
\usepackage[normalem]{ulem}
\usepackage{bbm}
\usepackage{mathabx}
\usepackage{comment}
\usepackage{float}

\usepackage[scr=euler]{mathalfa}

\def\mathclap#1{\text{\hbox to 0pt{\hss$\mathsurround=0pt#1$\hss}}}

\newcommand{\N}{\mathbb{N}}
\newcommand{\R}{\mathbb{R}}

\newcommand{\CH}{\mathcal{CH}^+}
\newcommand{\Hp}{\mathcal{H}^+}

\newcommand{\Mun}{\underline{\mathcal{M}}}

\newcommand{\Sp}{\mathbb{S}}

\newcommand{\vols}{\mathrm{vol}_{\mathbb{S}^2}}

\newcommand{\swl}{\mathring{\slashed{\Delta}}_{[s]}}

\newcommand{\un}{\underline}
\newcommand{\Pb}{\mathbb{P}_{S(l>2)}}
\newcommand{\loc}{\mathrm{loc}}

\newcommand{\f}{\nicefrac}

\newcommand{\rd}{\partial}

\begin{document}

\numberwithin{equation}{section}
\newtheorem{theorem}[equation]{Theorem}
\newtheorem{remark}[equation]{Remark}
\newtheorem{assumption}[equation]{Assumption}
\newtheorem{claim}[equation]{Claim}
\newtheorem{lemma}[equation]{Lemma}
\newtheorem{definition}[equation]{Definition}
\newtheorem{corollary}[equation]{Corollary}
\newtheorem{proposition}[equation]{Proposition}

\title{A note on the instability of the Kerr Cauchy horizon under linearised gravitational perturbations}

\author{Jan Sbierski\thanks{School of Mathematics, 
University of Edinburgh,
James Clerk Maxwell Building,
Peter Guthrie Tait Road, 
Edinburgh, 
EH9 3FD,
United Kingdom}}
\date{\today}

\maketitle

\begin{abstract}
This note slightly strengthens the result of \cite{Sbie23} on the linear instability of the Kerr Cauchy horizon.  This strengthened result is used in the proof \cite{LukSbie26} of the non-linear instability of the Kerr Cauchy horizon.
\end{abstract}

\section{Introduction}

This note shows how to slightly strengthen the result \cite{Sbie23}, which proves a blow-up result for the $s=+2$ Teukolsky equation at the Cauchy horizon in the interior of subextremal rotating Kerr black holes. Using slightly stronger assumptions on the asymptotic behaviour of the Teukolsky field along the event horizon (which are still satisfied by solutions arising from generic initial data posed on a global spacelike Cauchy hypersurface), we infer here slightly stronger and more detailed blow-up estimates near the Cauchy horizon. 

In \cite{Sbie24} a $C^{0,1}_{\loc}$-inextendibility result was proven for weak null singularities which are expected to form at the Cauchy horizon in the interior of generic vacuum black holes. In \cite{Sbie24}, \cite{Sbie24b} it was reported a while ago that the instability result of \cite{Sbie23} can be strengthened so that, \emph{if the same blow-up held at the non-linear level}, the assumptions of the $C^{0,1}_{\loc}$-inextendibility result \cite{Sbie24} would be met. This note provides the reported slight improvement. 

The  strategy just outlined for showing the formation of \emph{Lipschitz-inextendible} weak null singularities in the interior of generically perturbed rotating black holes is concluded in collaboration with Luk in \cite{LukSbie26} (relying on the works \cite{DafLuk17}, \cite{DafLuk26} of Dafermos and Luk), which uses the linear result of this note. We also refer the reader to the pointwise analysis \cite{Gur24} for the linear $s=+2$ Teukolsky field by Gurriaran, which, at the linear level is also strong enough to verify the curvature blow-up condition from \cite{Sbie24}. Furthermore, while finalising \cite{LukSbie26}, Gurriaran has posted the related paper \cite{Gur26} on the arXiv proving a very similar result.

The slight improvement over \cite{Sbie23} provided in this note is stated in Theorem \ref{MainThm}. The improvements are as follows: i) by adding the assumptions that  the angular $l=2$ mode of the $s=+2$ Teukolsky field  has the slowest decay on the event horizon and that higher angular modes as well as a time-derivative decay a little faster on the event horizon, we show that the corresponding conclusions on decay  hold near the Cauchy horizon; ii) the instability result is formulated on a general hypersurface which approaches the Cauchy horizon; iii) we allow for odd powers in the polynomial weights of the energy norms capturing the decay (and blow-up) of the Teukolsky field.

The proof of Theorem \ref{MainThm} is  a small modification of the proof in \cite{Sbie23}.   We use the convention that \underline{underlined} references to equations, theorems, etc.\ refer to those in \cite{Sbie23}. 
Moreover, we follow  the notation of \cite{Sbie23} throughout this note, all of which can be found in \underline{Section 2}.

\section{Main result}

We consider the interior of a subextremal Kerr black hole with parameters $0 < |a| < M$ in $(v_+, r, \theta, \varphi_+)$ coordinates, which are related to the Boyer-Lindquist coordinates $(t,r, \theta, \varphi)$ by $v_+ = t + r^*$ and $\varphi_+ = \varphi + \overline{r} \:\:  \mathrm{mod}\:  2\pi$, where $\frac{dr^*}{dr} = \frac{r^2 + a^2}{\Delta}$ and $\frac{d\overline{r}}{dr} = \frac{a}{\Delta}$ with $\Delta = r^2 - 2Mr + a^2$. The manifold with left event horizon $\Hp_l$ and right event horizon $\Hp_r$ attached is denoted by $(\underline{\mathcal{M}},g)$. We also introduce the functions $f^+ = v_+ -r + r_+$ and $f^- = -v_+ + 2r^* - r + r_-$, where $r_- < r_+$ are the roots of $\Delta$. In $(v_+,r, \theta, \varphi_+)$ coordinates, the Teukolsky equation takes the form
\begin{equation*}
\begin{split}
\mathcal{T}_{[s]} \psi_s := & a^2 \sin^2 \theta \,\partial_{v_+}^2\psi_s + 2a \,\partial_{v_+}\partial_{\varphi_+} \psi_s + 2(r^2 + a^2)\, \partial_{v_+}\partial_r \psi_s 
+2 a\, \partial_{\varphi_+}\partial_r \psi_s + \Delta \,\partial_r^2 \psi_s \\
&+ 2\Big( r(1-2s) - isa\cos \theta\Big)\, \partial_{v_+} \psi_s +2(r-M)(1-s) \,\partial_r \psi_s + \mathring{\slashed{\Delta}}_{[s]} \psi_s - 2s \psi_s 
= 0 \;,
\end{split}
\end{equation*}
where the spin $s$-weighted spherical Laplacian $\mathring{\slashed\Delta}_{[s]} = (\tilde{Z}_{1,+})^2 + (\tilde{Z}_{2,+})^2 + (\tilde{Z}_{3,+})^2 + s + s^2$ is given in terms of the vector fields
\begin{equation*}
\begin{aligned}
\tilde{Z}_{1,+} &= - \sin \varphi_+ \, \partial_\theta + \cos \varphi_+(-is \frac{1}{\sin \theta} - \frac{\cos \theta}{\sin \theta} \partial_{\varphi_+}) \\
\tilde{Z}_{2,+} &= -\cos \varphi_+ \, \partial_\theta - \sin \varphi_+(-is \frac{1}{\sin \theta} - \frac{\cos \theta}{\sin \theta} \partial_{\varphi_+}) \\
\tilde{Z}_{3,+} &= \partial_{\varphi_+} \;.
\end{aligned}
\end{equation*}
The spaces $\mathscr{I}^k_{[2]}(\Mun)$ of $C^k$-regular spin $2$-weighted functions are defined in Definition 2.74 of \cite{LukSbie26}, where $k \in \N_0$.
For $\psi \in \mathscr{I}^k_{[2]}(\Mun)$, $k \in \N_0$, we define the projection onto the $ml$ spin $2$-weighted spherical harmonic $Y_{ml}^{[2]}(\theta, \varphi_+; 0) = S_{ml}^{[2]}(\cos \theta;0) e^{im \varphi_+}$ (see \underline{Proposition 5.2})  by 
\begin{equation} \label{DefP}
(\mathbb{P}_{S(ml)} \psi )(v_+, r, \theta, \varphi_+) := \int_{\Sp^2} \psi(v_+,r,\theta', \varphi_+') \overline{Y_{ml}^{[2]}(\theta',\varphi_+';0)} \, \vols \cdot Y_{ml}^{[2]}(\theta, \varphi_+;0) \;.
\end{equation}
The subscript `S' always denotes the projection onto the spin $2$-weighted \emph{spherical} harmonics -- in contrast to the projection onto the spin $2$-weighted spheroidal harmonics. We also define $\psi_{S(l = 2)} := \sum_{m = -2}^2 \mathbb{P}_{S(m2)} \psi$ and $\psi_{S(l > 2)} := \psi - \psi_{S(l=2)}$. We can now state the main theorem proven in this note.

\begin{theorem} \label{MainThm}
Consider a patch of $\Mun$ given by $\Mun \cap \{f_- \leq v_1\} \cap \{f_+ \geq v_0\}$ for some $v_0 >1$, $ v_1 \in \R$, see also \underline{Figure 9}. Let $\psi \in \mathscr{I}^{10}_{[2]}(\Mun\cap \{f_- \leq v_1\} \cap \{f_+ \geq v_0\} )$ satisfy the Teukolsky equation $\mathcal{T}_{[2]}\psi = 0$. Let $q \in \N_{\geq 4}$ and assume that $\psi$ satisfies for any $ q_- <q$ close to $q$
\begin{align}
&\int\limits_{\Hp_r \cap \{ v_+ \geq v_0\}} v_+^{q} |\psi_{S(l=2)}|^2 \, \vols dv_+ = \infty \label{EqThmA} \\
\sum_{0\leq i_1+i_2+i_3+j \leq 1} &\int\limits_{\Hp_r \cap \{v_+ \geq v_0\}} v_+^{q_-} |\widetilde Z_{1,+}^{i_1}\widetilde Z_{2,+}^{i_2}\widetilde Z_{3,+}^{i_3}\rd_{v_+}^{j} f|^2 \,\vols\,dv_{+} < +\infty \label{EqThmB} \\
\sum_{0\leq i_1+i_2+i_3+j \leq 1} &\int\limits_{\Hp_r \cap \{v_+ \geq v_0\}} v_+^{q } |\widetilde Z_{1,+}^{i_1}\widetilde Z_{2,+}^{i_2}\widetilde Z_{3,+}^{i_3}\rd_{v_+}^{j+1} f|^2 \,\vols\,dv_{+} < +\infty \label{EqThmC} \\
\sum_{0\leq i_1+i_2+i_3+j \leq 1} &\int\limits_{\Hp_r \cap \{v_+ \geq v_0\}} v_+^{q} |\widetilde Z_{1,+}^{i_1}\widetilde Z_{2,+}^{i_2}\widetilde Z_{3,+}^{i_3}\rd_{v_+}^{j} f_{S(l>2)}|^2 \,\vols\,dv_{+} < +\infty \label{EqThmD} 
\end{align}
with $f \in \{\rd_{v_+}^a\rd_{\varphi_+}^b\rd_r^c\psi, \rd_{v_+} \rd_{v_+}^a\rd_{\varphi_+}^b\rd_r^c\psi, \mathcal{Q}_{[2]} \rd_{v_+}^a\rd_{\varphi_+}^b\rd_r^c \psi, \}, \;0 \leq a + b \leq 2, \;c = 0,1,2 $.\footnote{The number of derivatives assumed here is not sharp and can be improved. Furthermore, $\mathcal{Q}_{[s]} := a^2 \sin^2 \theta \, \rd_{v_+}^2 - 2isa\cos \theta \, \rd_{v_+} + \swl$ denotes the spin $s$-weighted Carter operator and $\vols = \sin \theta d \theta d \varphi_+$.} 

Let $\Sigma$ be a hypersurface given as a graph $\Sigma = \{ \big(v_+, r_\Sigma(v_+, \theta, \varphi_+), \theta, \varphi_+ \big) \in \Mun \; | \; (v_+, \theta, \varphi_+) \in (v_0, \infty) \times \Sp^2\}$, where $r_\Sigma : (v_0, \infty) \times \Sp^2 \to (r_-, r_+)$ is a smooth function with $\sup_{(\theta, \varphi_+) \in \Sp^2} |r_\Sigma(v_+, \theta, \varphi_+) - r_-| \lesssim v_+^{- \sigma}$ for some $\sigma >0$.
Then the following bounds hold:

\begin{minipage}{0.35\textwidth}
\begin{flalign}
&\int\limits_\Sigma v_+^{q} | \psi_{S(l = 2)}|^2 \, \vols dv_+ = \infty &\label{EqThmConA} \\
&\int\limits_\Sigma v_+^{q_-} | \psi|^2 \, \vols dv_+ < \infty &\label{EqThmConB} 
\end{flalign}
\end{minipage}
\begin{minipage}{0.12\textwidth}
${}$
\end{minipage}
\begin{minipage}{0.49\textwidth}
\begin{align}
&\int\limits_\Sigma v_+^{q} |\rd_{v_+} \psi|^2 \, \vols dv_+ < \infty \label{EqThmConC} \\
&\int\limits_\Sigma v_+^{q} | \psi_{S(l > 2)}|^2 \, \vols dv_+ < \infty \label{EqThmConD} \;.
\end{align}
\end{minipage}
\end{theorem}
In particular, \eqref{EqThmConA} and \eqref{EqThmConD} are the statements needed in \cite{LukSbie26}. If in assumptions \eqref{EqThmC}, \eqref{EqThmD} the weight $q$ is replaced by $q + \lambda$ with $\lambda \geq 0$, then this faster decay can also be propagated to $\Sigma$ in \eqref{EqThmConC} and \eqref{EqThmConD}.

To understand the structure of the assumptions made on the event horizon, we refer the reader to \underline{Remark 3.11} on the asymptotics of the Teukolsky field on the event horizon from \cite{MaZha23}.



\section{Auxiliary results}

\subsection{Some elementary results concerning fractional Sobolev spaces}

In this section we collect a few basic results concerning fractional Sobolev spaces that are needed in this paper. We only need the $L^2$-based fractional Sobolev spaces and only in one dimension. We use here the Gagliardo approach which localises naturally. A detailed introduction to fractional Sobolev spaces can be found in \cite{FracSobolev}.

\begin{definition}
Let $\Omega \subseteq \R$ be open and connected. We define for $k \in \N_0$
\begin{enumerate}
\item $H^{k}(\Omega) := \{ f \in L^2(\Omega) \; | \; \rd^j f \in L^2(\Omega) \textnormal{ for all } 0 \leq j \leq k\}$
\item $H^{\nicefrac{1}{2}}(\Omega) := \{ f \in L^2(\Omega) \; | \; \frac{f(\omega) - f(\xi)}{|\omega-\xi|} \in L^2(\Omega \times \Omega)\}$
\item $H^{k + \nicefrac{1}{2}}(\Omega) := \{ f \in H^{k}(\Omega) \; | \; \rd^kf \in H^{\nicefrac{1}{2}}(\Omega) \}$ \;.
\end{enumerate}
Here, all derivatives are weak derivatives. For $f \in H^{\f{1}{2}}(\Omega)$ we set $$[f]_{H^{\f{1}{2}}(\Omega)}^2 := \iint\limits_{\Omega \times \Omega} \frac{|f(\omega) - f(\xi)|^2}{|\omega-\xi|^2} \, d\omega d \xi \quad \textnormal{ and } \quad ||f||_{H^{\nicefrac{1}{2}}(\Omega)}^2 := ||f||_{L^2(\Omega)}^2 + [f]_{H^{\f{1}{2}}(\Omega)}^2\;.$$
\end{definition}

\begin{lemma}\label{Lem1contained}
We have $H^{1}(\Omega) \subseteq H^{\nicefrac{1}{2}}(\Omega)$.
\end{lemma}

\begin{lemma}\label{LemRest}
Let $\Omega \subseteq V \subseteq \R$ and $f \in H^{k + \nicefrac{1}{2}}(V)$. Then $f|_\Omega \in H^{k + \nicefrac{1}{2}}(\Omega)$.
\end{lemma}

\begin{lemma}\label{LemChi1}
Let $f \in H^{\nicefrac{1}{2}}(\Omega)$ and $ \chi \in C^1(\Omega)$ with $\sum_{j = 0}^{1} ||\rd^j \chi ||_{L^\infty(\Omega)} < \infty$. Then $\chi \cdot f \in H^{\f{1}{2}}(\Omega)$.
\end{lemma}
For the proof of Lemma \ref{Lem1contained} see Proposition 2.2 in \cite{FracSobolev}, Lemma \ref{LemRest} is immediate, and Lemma \ref{LemChi1} follows from  Lemma 5.3 in \cite{FracSobolev}.

\begin{lemma}\label{LemChiK}
Let $f \in H^{k+\f{1}{2}}(\Omega)$ and $\chi \in C^{k+1}(\Omega)$ with $\sum_{j = 0}^{k+1}||\rd^j \chi||_{L^\infty(\Omega)} < \infty$. Then $\chi \cdot f \in H^{k + \f{1}{2}}(\Omega)$.
\end{lemma}
\begin{proof}
For $i \leq k$ we have $\rd^i(\chi f) = \sum_{j = 0}^i \binom{i}{j} \rd^j \chi \cdot \rd^{i - j} f \in L^2(\Omega)$, since $\rd^j \chi$ is uniformly bounded. Moreover, $\rd^k(\chi f) = \sum_{j =1}^k \binom{k}{j} \rd^j \chi \cdot \rd^{k - j}f + \chi \cdot \rd^k f$. For the last term we use Lemma \ref{LemChi1} to infer that it is in $H^{\f{1}{2}}(\Omega)$, while for the first term we differentiate once more and use Lemma \ref{Lem1contained}.
\end{proof}

\begin{corollary}\label{CorChi}
Let $\Omega \Subset \R$ be precompact, $\chi \in C^\infty(\overline{\Omega})$, and $|\chi| \geq c > 0$. Then $f \in H^{\f{k}{2}}(\Omega)$ if, and only if, $\chi \cdot f \in H^{\f{k}{2}}(\Omega)$ for $k \in \N_0$.
\end{corollary}

\begin{proof}
This follows directly from Lemma \ref{LemChiK} by also using multiplication by $\frac{1}{\chi}$ for the reverse direction.
\end{proof}

The following shows that in the case of $\Omega = \R$ the spaces $H^{k+\f{1}{2}}(\Omega)$ defined via the Gagliardo semi-norm coincide with the definition of fractional Sobolev spaces via the Fourier transform.
\begin{proposition} \label{PropFourierSob}
It holds that $H^{k + \f{1}{2}}(\R) = \{ f \in L^2(\R) \; | \; \int_{\R} (1 + |v|^{2k+1}) |\check{f}(v)|^2 \, d v < \infty\}$. Here $\check{f}$ denotes the Fourier transform of $f$. Moreover, there exists $C>0$ such that
\begin{equation*}
[f]_{H^{\f{1}{2}}(\R)}^2 = C \int_{\R} |v| |\check{f}(v)|^2 dv \;.
\end{equation*}
\end{proposition}
This follows directly from Proposition 3.4 in \cite{FracSobolev}.

\begin{proposition} \label{PropLocH}
Let $h \in L^2(\R)$ with $h \notin H^{\f{1}{2}}(\R)$ and $\omega \cdot h \in H^{\frac{1}{2}}(\R)$. Then $h \notin H^{\f{1}{2}}\big((-\varepsilon, \varepsilon)\big)$ for any $\varepsilon >0$.
\end{proposition}

\begin{proof}
The second assumption in the proposition reads $\int_\R \int_\R \frac{|h(\omega) - h(\xi)|^2}{|\omega - \xi|^2} \, d \omega d\xi = \infty$.  We now localise the divergence. First observe that we have for all $\delta >0$
\begin{equation}
\label{EqPfH121}
\begin{split}
\int\int_{|\omega - \xi| > \delta} \frac{|h(\omega) - h(\xi)|^2}{|\omega - \xi|^2} \, d \omega d\xi  & \leq \int\int_{|\omega - \xi| > \delta} \frac{2|h(\omega)|^2 + 2| h(\xi)|^2}{|\omega - \xi|^2} \, d \omega d\xi \\
&= 4\int_{|\chi > \delta|} \int_\R \frac{|h(\omega)|^2}{|\chi|^2} \, d\omega d \chi \\
 &< \infty \;.
\end{split}
\end{equation}
Furthermore, note that
\begin{equation*}
\begin{split}
\omega^2 | h(\omega) - h(\xi)|^2 &= |\omega \cdot h(\omega) - \xi \cdot h(\xi) + \xi \cdot h(\xi) - \omega \cdot h(\xi)|^2 \leq 2 \Big[ | \omega \cdot h(\omega) - \xi \cdot h(\xi)|^2 + |\xi - \omega|^2 \cdot |h(\xi)|^2 \Big] \;. 
\end{split}
\end{equation*}
Hence, we obtain
\begin{equation*}
\begin{split}
\iint\limits_{\substack{|\omega - \xi| < \delta \\ |\omega + \xi| > \delta}} \frac{|h(\omega) - h(\xi)|^2}{|\omega - \xi|^2} \, d \omega d\xi  &\leq c(\delta) \iint\limits_{\substack{|\omega - \xi| < \delta \\ |\omega + \xi| > \delta}} \frac{(\omega^2 + \xi^2) |h(\omega) - h(\xi)|^2}{|\omega - \xi|^2} \, d \omega d\xi \\
&\leq c(\delta) \iint\limits_{\substack{|\omega - \xi| < \delta \\ |\omega + \xi| > \delta}} \Big(\frac{4|\omega \cdot h(\omega) - \xi \cdot h(\xi)|^2}{|\omega - \xi|^2}  + 2 \big(|h(\xi)|^2 + |h(\omega)|^2\big) \Big)\, d \omega d\xi \\
&\leq 4 c(\delta) \Big( [\omega \cdot h]_{H^{\f{1}{2}}(\R)} + \iint\limits_{\substack{|\omega - \xi| < \delta \\ |\omega + \xi| > \delta}} |h(\xi)|^2 \, d\omega d \xi \Big) \\
&< \infty \;.
\end{split}
\end{equation*}
\noindent %
\begin{minipage}{0.48\textwidth}
Together with \eqref{EqPfH121} this gives for all $\delta >0$
\[
\iint\limits_{\substack{|\omega - \xi| < \delta \\ |\omega + \xi| < \delta}} \frac{|h(\omega) - h(\xi)|^2}{|\omega - \xi|^2} \, d \omega d\xi = \infty \;, 
\]
from which the proposition follows.
\end{minipage}%
\hfill 
\begin{minipage}{0.48\textwidth}
\centering
  \def\svgwidth{4cm} 
\begingroup%
  \makeatletter%
  \providecommand\color[2][]{%
    \errmessage{(Inkscape) Color is used for the text in Inkscape, but the package 'color.sty' is not loaded}%
    \renewcommand\color[2][]{}%
  }%
  \providecommand\transparent[1]{%
    \errmessage{(Inkscape) Transparency is used (non-zero) for the text in Inkscape, but the package 'transparent.sty' is not loaded}%
    \renewcommand\transparent[1]{}%
  }%
  \providecommand\rotatebox[2]{#2}%
  \newcommand*\fsize{\dimexpr\f@size pt\relax}%
  \newcommand*\lineheight[1]{\fontsize{\fsize}{#1\fsize}\selectfont}%
  \ifx\svgwidth\undefined%
    \setlength{\unitlength}{419.62340828bp}%
    \ifx\svgscale\undefined%
      \relax%
    \else%
      \setlength{\unitlength}{\unitlength * \real{\svgscale}}%
    \fi%
  \else%
    \setlength{\unitlength}{\svgwidth}%
  \fi%
  \global\let\svgwidth\undefined%
  \global\let\svgscale\undefined%
  \makeatother%
  \begin{picture}(1,0.77415652)%
    \lineheight{1}%
    \setlength\tabcolsep{0pt}%
    \put(0,0){\includegraphics[width=\unitlength,page=1]{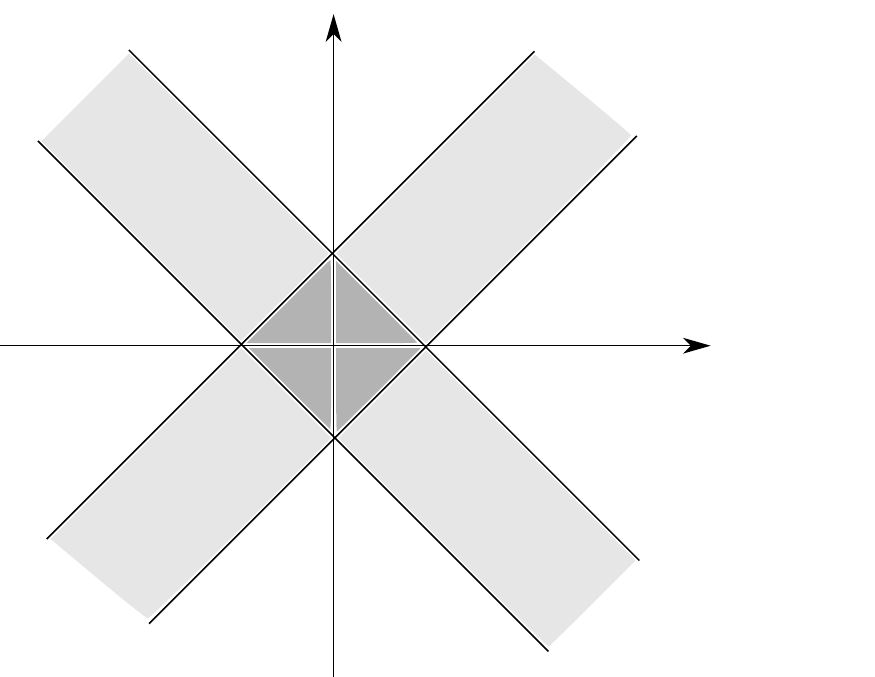}}%
    \put(0.80078481,0.320892){\color[rgb]{0,0,0}\makebox(0,0)[lt]{\lineheight{1.25}\smash{\begin{tabular}[t]{l}$\omega$\end{tabular}}}}%
    \put(0.40761971,0.71983901){\color[rgb]{0,0,0}\makebox(0,0)[lt]{\lineheight{1.25}\smash{\begin{tabular}[t]{l}$\xi$\end{tabular}}}}%
    \put(0,0){\includegraphics[width=\unitlength,page=2]{H12.pdf}}%
    \put(0.78922098,0.52036551){\color[rgb]{0,0,0}\makebox(0,0)[lt]{\lineheight{1.25}\smash{\begin{tabular}[t]{l}$|\omega - \xi| < \delta$\end{tabular}}}}%
    \put(0.76898461,0.17345506){\color[rgb]{0,0,0}\makebox(0,0)[lt]{\lineheight{1.25}\smash{\begin{tabular}[t]{l}$|\omega + \xi| < \delta$\end{tabular}}}}%
    \put(0,0){\includegraphics[width=\unitlength,page=3]{H12.pdf}}%
  \end{picture}%
\endgroup%

   Domains of integration
\end{minipage}
\end{proof}

\subsection{Relating $\check{\psi}_{ml}(\omega)$ to $\check{\psi}_{S(ml)}(\omega)$ near $\omega = 0$}

\begin{proposition} \label{PropRelateSNew}
Let $h \in L^2_{v_+}L^2_{\Sp^2}$ and $1 \leq {\mathfrak{p}} \in \N$ such that 
\begin{equation}\label{PropA}
\int\limits_{-\infty}^\infty \int\limits_{\Sp^2} (1 + |v_+|^{2({\mathfrak{p}}-1)}) |h(v_+, \theta, \varphi_+)|^2 \, \vols dv_+ < \infty
\end{equation}
and
\begin{equation}
\label{PropB}
\int\limits_{-\infty}^\infty \int\limits_{\Sp^2} (1 + |v_+|^{2\mathfrak{p}-1}) |\rd_{v_+}h(v_+, \theta, \varphi_+)|^2 \, \vols dv_+ < \infty
\end{equation}
hold. Then for all $|m_0| \leq l_0$, $l_0 \geq 2$ there exists $\delta_0(m_0,l_0) >0$ such that if  $\delta_0 > \delta >0$, we have $\check{h}_{S(m_0l_0)} \in H^{{\mathfrak{p}} - \f{1}{2}}\big((-\delta, \delta)\big)$ if, and only if, $\check{h}_{m_0l_0} \in H^{{\mathfrak{p}} - \f{1}{2}}\big((- \delta, \delta) \big)$.
\end{proposition}

\begin{proof}
The proof uses similar techniques to those used in the proof of \underline{Proposition 7.5}. We first note that
$$\int\limits_{-\infty}^\infty \int\limits_{\Sp^2} (1 + |v_+|^{2p-1}) |\rd_{v_+}h(v_+, \theta, \varphi_+)|^2 \, \vols dv_+ = \sum_{m,l}\int\limits_{-\infty}^\infty (1 + |v_+|^{2p-1}) |\rd_{v_+} h_{S(ml)}(v_+)|^2 \, dv_+  $$
so that \eqref{PropA} and \eqref{PropB} imply
\begin{equation}
\label{EqPropA}
\int\limits_{-\infty}^\infty \sum_{m,l} |\rd_\omega^{\mathfrak{q}} \check{h}_{S(ml)}|^2 \, d \omega < \infty \quad \textnormal{ for } 0 \leq {\mathfrak{q}} \leq {\mathfrak{p}}-1
\end{equation}
and, by Proposition  \ref{PropFourierSob},
\begin{equation} \label{EqPropNew}
 \sum_{m,l}\int\limits_{-\infty}^\infty |\rd_\omega^{\mathfrak{q}} (\omega \check{h}_{S(ml)})|^2 \, d \omega < \infty  \textnormal{ for } 0 \leq {\mathfrak{q}} \leq {\mathfrak{p}}-1 \textnormal{ and } \sum_{m,l} \big[\rd_\omega^{{\mathfrak{p}}-1} (\omega  \check{h}_{S(ml)})\big]_{H^{\nicefrac{1}{2}}(\R)}^2 < \infty \;.
\end{equation}

Using $\rd_\omega^{\mathfrak{q}}(\omega \check{h}_{S(ml)}) = {\mathfrak{q}} \cdot \rd_\omega^{{\mathfrak{q}}-1} \check{h}_{S(ml)} + 
\omega \rd_\omega^{\mathfrak{q}} \check{h}_{S(ml)}$ we obtain from \eqref{EqPropA}, \eqref{EqPropNew}, and Lemma \ref{Lem1contained} in particular
\begin{equation}
\label{EqPropB}
 \sum_{m,l} \big[\omega \rd_\omega^{{\mathfrak{p}}-1}  \check{h}_{S(ml)}\big]^2_{H^{\nicefrac{1}{2}}(\R)} < \infty  \;.
\end{equation}
Recall now the $\omega$-dependent change of orthonormal basis in $L^2([-1,1], d \cos \theta)$
$$S_{ml}^{[s]}(\cos \theta; \omega) = \underbrace{\int_{[-1,1]} S_{ml}^{[s]}(\cos \theta; \omega) S_{ml'}^{[s]}(\cos \theta; 0) \, d \cos \theta}_{=:E^{[s]}_{mll'}(\omega)} \;\cdot\, S^{[s]}_{ml'}(\cos \theta; 0) \;,$$
and that $\sup_{|\omega| \leq \omega_0}\sum_{l'} |\rd_\omega^{\mathfrak{q}} E_{mll'}^{[s]}(\omega)|^2  \leq C(\omega_0, m,l,{\mathfrak{q}})$ by the smoothness of $S_{ml}^{[s]}(\cos \theta; \omega)$ in $\omega$, see \underline{(7.12)}. We have 
\begin{equation} \label{EqHML}
\check{h}_{m_0 l_0}(\omega) = \sum_l  E_{m_0l_0l}^{[s]} (\omega)   \cdot \check{h}_{S(m_0l)}(\omega)
\end{equation} 
in $ L^2_\omega(\R)$. We first consider the terms with $l \neq l_0$ and claim that
\begin{equation}
\label{EqPropClaim}
\sum_{l \neq l_0} E_{m_0l_0l}^{[s]} (\omega)   \cdot  \widecheck{h}_{S(m_0l)}(\omega) \in H^{{\mathfrak{p}}- \frac{1}{2}}\big((-\delta, \delta)\big)
\end{equation}
for each $\delta >0$. 
To prove this, we  compute for $0 \leq {\mathfrak{q}} \leq {\mathfrak{p}}-1$
\begin{equation} \label{EqStructure}\rd_\omega^{\mathfrak{q}} \sum_{l \neq l_0} E_{m_0l_0l}^{[s]} (\omega)   \cdot  \widecheck{h}_{S(m_0l)}(\omega) = \sum_{l \neq l_0} \sum_{{\mathfrak{q}}' = 0}^{\mathfrak{q}} \binom{{\mathfrak{q}}}{{\mathfrak{q}}'} \rd_\omega^{{\mathfrak{q}} - {\mathfrak{q}}'} E_{m_0l_0l}^{[s]} (\omega) \cdot \rd_\omega^{{\mathfrak{q}}'} \widecheck{h}_{S(m_0l)}(\omega) \;.
\end{equation}
For ${\mathfrak{q}}' \leq {\mathfrak{p}}-1$ we estimate
$$ \int\limits_{(-\delta, \delta)} \Big| \sum_{l \neq l_0}\rd_\omega^{{\mathfrak{q}} - {\mathfrak{q}}'} E_{m_0l_0l}^{[s]} (\omega) \cdot \rd_\omega^{{\mathfrak{q}}'} \widecheck{h}_{S(m_0l)}(\omega)\Big|^2 \, d \omega \leq \int\limits_{(-\delta, \delta)} \underbrace{\Big( \sum_{l \neq l_0} |\rd_\omega^{{\mathfrak{q}} - {\mathfrak{q}}'} E_{m_0l_0l}^{[s]} (\omega)|^2\Big)}_{\leq C(\delta, m_0, l_0,{\mathfrak{q}}-{\mathfrak{q}}')} \Big( \sum_{l \neq l_0} |\rd_\omega^{{\mathfrak{q}}'} \widecheck{h}_{S(m_0l)}(\omega)|^2\Big) \, d \omega < \infty\;, $$
where we have used \eqref{EqPropA}. 
We now need to show that $\rd_\omega^{{\mathfrak{p}}-1} \sum_{l \neq l_0} E_{m_0l_0l}^{[s]} (\omega)   \cdot  \widecheck{h}_{S(m_0l)}(\omega) \in H^{ \frac{1}{2}}\big((-\delta, \delta)\big)$. The structure of this expression is given in \eqref{EqStructure} with ${\mathfrak{q}} = {\mathfrak{p}}-1$. Consider the terms with ${\mathfrak{q}}' < {\mathfrak{p}}-1$. For those we can take another derivative and repeat the above argument to show that they are even in $H^{1}\big((-\delta, \delta)\big)$. So it remains to show that 
\begin{equation} \label{EqNewTerm}
\sum_{l \neq l_0}E_{m_0l_0l}^{[s]} (\omega)   \cdot \rd_{\omega}^{{\mathfrak{p}}-1} \widecheck{h}_{S(m_0l)}(\omega) \in  H^{ \frac{1}{2}}\big((-\delta, \delta)\big)\;.
\end{equation}
We use that $E_{m_0l_0l}^{[s]}(0) =0$ for $l \neq l_0$ and thus 
\begin{equation}
\label{EqBoundOneOverOmega}
\sum_{l \neq l_0}\Big| \frac{1}{\omega} E_{m_0l_0l}^{[s]} (\omega)\Big|^2 = \sum_{l \neq l_0}\Big| \frac{1}{\omega}  \int_0^\omega \rd_{\omega'} E_{m_0l_0l}^{[s]}(\omega') \, d \omega'\Big|^2 \leq C
\end{equation}
uniformly for $\omega \in (-\delta, \delta)$ by the same computation as below $\underline{(7.14)}$. Furthermore, we note that
$$\frac{\xi}{\omega}  E_{m_0l_0l}^{[s]} (\omega) - E_{m_0l_0l}^{[s]} (\xi) = \big( \frac{\xi}{\omega} - 1\big) E_{m_0l_0l}^{[s]} (\omega) + E_{m_0l_0l}^{[s]} (\omega) - E_{m_0l_0l}^{[s]} (\xi) = (\xi - \omega) \frac{1}{\omega} E_{m_0l_0l}^{[s]} (\omega) + \int_{\xi}^\omega \rd_{\omega'} E_{m_0l_0l}^{[s]} (\omega') \, d\omega'$$
and thus
\begin{equation}\label{EqSecondBoundNew}
\begin{split}
\frac{1}{2}\sum_{l \neq l_0}\Big|\frac{\xi}{\omega}  E_{m_0l_0l}^{[s]} (\omega) - E_{m_0l_0l}^{[s]} (\xi) \Big|^2  &\leq \sum_{l \neq l_0} |\xi - \omega|^2 \Big| \frac{1}{\omega} E_{m_0l_0l}^{[s]} (\omega)\Big|^2 + \sum_{l \neq l_0} \Big| \int_{\xi}^\omega \rd_{\omega'} E_{m_0l_0l}^{[s]} (\omega') \, d\omega' \Big|^2 \\
&\leq |\xi - \omega|^2 C +  \sum_{l \neq l_0} |\omega - \xi| \int_{\xi}^\omega | \rd_{\omega'} E_{m_0l_0l}^{[s]} (\omega') |^2\, d\omega' \leq 2|\xi - \omega|^2 C \;,
\end{split}
\end{equation}
where we have used \eqref{EqBoundOneOverOmega}. Using \eqref{EqBoundOneOverOmega} and \eqref{EqSecondBoundNew} we now prove \eqref{EqNewTerm}:
\begin{equation*}
\begin{split}
\frac{1}{2}\int\limits_{(-\delta, \delta)^2}& \frac{ \Big| \sum_{l \neq l_0}E_{m_0l_0l}^{[s]} (\omega)   \cdot \rd_{\omega}^{{\mathfrak{p}}-1} \widecheck{h}_{S(m_0l)}(\omega) - \sum_{l \neq l_0}E_{m_0l_0l}^{[s]} (\xi)   \cdot \rd_{\xi}^{{\mathfrak{p}}-1} \widecheck{h}_{S(m_0l)}(\xi) \Big|^2}{|\omega - \xi|^2} \, d \omega d \xi \\
&\leq \int\limits_{(-\delta, \delta)^2} \frac{\Big|\sum_{l \neq l_0} \frac{E_{m_0l_0l}^{[s]} (\omega) }{\omega} \Big( \omega  \rd_{\omega}^{{\mathfrak{p}}-1} \widecheck{h}_{S(m_0l)}(\omega) - \xi  \rd_{\xi}^{{\mathfrak{p}}-1} \widecheck{h}_{S(m_0l)}(\xi)\Big)|^2}{|\omega - \xi|^2} \, d \omega d \xi \\
&\quad +  \int\limits_{(-\delta, \delta)^2} \frac{\Big|\sum_{l \neq l_0} \xi \rd_{\xi}^{{\mathfrak{p}}-1} \widecheck{h}_{S(m_0l)}(\xi) \Big( \frac{ E_{m_0l_0l}^{[s]} (\omega) }{\omega} -  \frac{ E_{m_0l_0l}^{[s]} (\xi) }{\xi} \Big) \Big|^2}{|\omega - \xi|^2} \, d \omega d \xi\\
&\leq  \int\limits_{(-\delta, \delta)^2} \sum_{l \neq l_0}\Big| \frac{1}{\omega} E_{m_0l_0l}^{[s]} (\omega)\Big|^2 \cdot \frac{ \sum_{l \neq l_0} \Big|  \omega  \rd_{\omega}^{{\mathfrak{p}}-1} \widecheck{h}_{S(m_0l)}(\omega) - \xi  \rd_{\xi}^{{\mathfrak{p}}-1} \widecheck{h}_{S(m_0l)}(\xi)\Big|^2}{|\omega - \xi|^2} \, d \omega d\xi \\
&\quad + \int\limits_{(-\delta, \delta)^2}   \sum_{l \neq l_0} | \rd_{\xi}^{{\mathfrak{p}}-1} \widecheck{h}_{S(m_0l)}(\xi)|^2 \cdot \frac{ \sum_{l \neq l_0}\Big|\frac{\xi}{\omega}  E_{m_0l_0l}^{[s]} (\omega) - E_{m_0l_0l}^{[s]} (\xi)\Big|^2}{|\omega - \xi|^2} \, d \omega d \xi <\infty\;,
\end{split}
\end{equation*}
where we have also used  \eqref{EqPropB} and \eqref{EqPropA}. This finally establishes \eqref{EqPropClaim}.

We thus obtain 
\begin{equation} \label{PropPfE} 
\widecheck{h}_{m_0 l_0}(\omega) = E_{m_0l_0l_0}^{[s]} (\omega)   \cdot \widecheck{h}_{S(m_0l_0)}(\omega) + \Big(\textnormal{terms in } H^{{\mathfrak{p}}- \frac{1}{2}}\big((-\delta, \delta)\big) \Big)\;.\end{equation} 
Now, if $\widecheck{h}_{S(m_0l_0)} \in H^{{\mathfrak{p}} - \f{1}{2}}\big((-\delta, \delta) \big)$, then the smoothness of $E_{m_0l_0l_0}^{[s]}$, together with Lemma \ref{LemChiK}, imply $\widecheck{h}_{m_0 l_0} \in H^{{\mathfrak{p}} - \f{1}{2}}\big((-\delta, \delta) \big)$. For the reverse direction we observe that there exists $\delta_0(m_0,l_0)>0$ such that $E_{m_0l_0l_0}^{[s]}(\omega) > \frac{1}{2}$ for all $|\omega| < \delta_0$. This follows from the smoothness of $E_{m_0l_0l_0}^{[s]}$ together with $E_{m_0l_0l_0}^{[s]}(0) = 1$. The result then follows from Corollary \ref{CorChi}.
\end{proof}

\section{Proof of Theorem \ref{MainThm}}

We begin by pointing out that the results in \cite{Sbie23} are stated for smooth spin $2$-weighted functions $\mathscr{I}^\infty_{[2]}(\Mun)$, but that this is nowhere needed -- finite regularity suffices.

We now apply \underline{Theorem 8.6} to extend $\psi \in \mathscr{I}^{10}_{[2]}(\Mun\cap \{f_- \leq v_1\} \cap \{f_+ \geq v_0\} )$ to a  spin $2$-weighted function on all of $\Mun$, denoted again by $\psi \in \mathscr{I}^8_{[2]}(\Mun)$, such that $\psi$ satisfies \underline{Assumption 2.46} and \underline{condition (3.5)}. Note that since $\psi$ satisfies \underline{Assumption 2.46} and $\rd_{v_+}$ is a smooth vector field on $\Mun$ (see \underline{A.1}) which commutes with the Teukolsky operator, we also have that $\rd_{v_+} \psi$ satisfies \underline{Assumption 2.46}. Furthermore, since $\rd_{v_+}$ is tangent to $\Hp_l$, $\rd_{v_+} \psi$ also satisfies \underline{condition (3.5)}. We first propagate the desired bounds to the left Cauchy horizon, i.e., we show

\begin{minipage}{0.45\textwidth}
\begin{flalign}
&\int\limits_{\CH_l \cap \{v_+ \geq 1\}} v_+^{q} | \psi_{S(l = 2)}|^2 \, \vols dv_+ = \infty \label{EqThmConA'} & \\
&\int\limits_{\CH_l \cap \{v_+ \geq 1\}} v_+^{q_-} | \psi|^2 \, \vols dv_+ < \infty & \label{EqThmConB'} 
\end{flalign}
\end{minipage}
\begin{minipage}{0.05\textwidth}
${}$
\end{minipage}
\begin{minipage}{0.45\textwidth}
\begin{align}
&\int\limits_{\CH_l \cap \{v_+ \geq 1\}} v_+^{q} |\rd_{v_+} \psi|^2 \, \vols dv_+ < \infty \label{EqThmConC'} \\
&\int\limits_{\CH_l \cap \{v_+ \geq 1\}} v_+^{q} | \psi_{S(l > 2)}|^2 \, \vols dv_+ < \infty \label{EqThmConD'} \;.
\end{align}
\end{minipage}

We observe that all results proved in \underline{Section 4} only depend on \underline{Assumption 2.46}, \underline{(3.4)}, and \underline{(3.5)} (or \underline{(3.6)}). Given our assumption \eqref{EqThmB}, we can set $q_r = q_-$ in \underline{(3.4)} and $q_l \geq q $ to obtain all results in \underline{Section 4} with those new values for $q_r$ and $q_l$. In particular, \un{Proposition 4.64} gives \eqref{EqThmConB'}.

Similarly, since $\rd_{v_+} \psi$ satisfies \un{Assumption 2.46} and \un{(3.5)}, and by our assumption \eqref{EqThmC}, also \un{(3.4)} with $q_r = q$, we also obtain all results in \un{Section 4} with $\psi$ replaced by $\rd_{v_+} \psi$ and the new value for $q_r$. Again, \un{Proposition 4.64} gives \eqref{EqThmConC'}. For later reference, we also take note
that we have
\begin{align}
&\textnormal{\un{Proposition 4.64} with $\psi$ replaced by $\rd_{v_+} \psi$ and $q_r = q$ and $q_l \geq q$} \label{EqV4.64}  \\
&\textnormal{\un{(4.19)} with $\psi$ replaced by $\rd_{v_+} \psi$ and $q_r = q$} \label{EqV4.19}  \\
&\textnormal{\un{(4.32)} with $\psi$ replaced by $\rd_{v_+} \psi$ and $q_r = q$} \label{EqV4.32}  \\
&\textnormal{\un{(4.52)} with $\psi$ replaced by $\rd_{v_+} \psi$ and $q_r = q$} \label{EqV4.52} 
\end{align}

We now prove \eqref{EqThmConD'}. \un{Proposition 4.1} (which only depends on Assumption \un{2.46} and \un{(3.5)} (or \un{(3.6)}) is applied unchanged to $\psi$ with $q_l \geq q$. We claim that \un{Proposition 4.11} holds indeed with $\psi$ replaced by $\mathbb{P}_{S(l>2)} \psi$ and $q_r = q$. This, however, is non-trivial, since $\Pb$ does not commute with the Teukolsky operator. 

Using that $\psi \in \mathscr{I}^8_{[2]}(\Mun)$, it is immediate from \eqref{DefP} that $\Pb$ commutes with $\rd_{v_+}, \rd_r, \rd_{\varphi_+}, \swl$ and multiplication by functions in $r$. We thus infer from \un{(2.39)} that $\Pb \psi$ satisfies the following inhomogeneous Teukolsky equation
\begin{equation}\label{EqComEq}
0 = \mathcal{T}_{[2]} \Pb \psi - a^2 [ \sin^2 \theta, \Pb] \, \rd_{v_+}^2 \psi + 2isa [ \cos \theta, \Pb ] \, \rd_{v_+} \psi \;.
\end{equation}
Note that the inhomogeneities all have at least one $\rd_{v_+}$ derivative and thus, as we have already shown, decay slightly better than $\psi$. This is what allows us to propagate the slightly improved decay of $\Pb \psi$. We flesh out the argument in the following.
First note that for $h \in L^2(\Sp^2)$ we have
\begin{equation} \label{EqCommutator}
\begin{split}
|| [ \sin^2 \theta, \Pb] \, h ||_{L^2(\Sp^2)} &\leq 2 ||h||_{L^2(\Sp^2)} \\
|| [ \cos \theta, \Pb] \, h ||_{L^2(\Sp^2)} &\leq 2 ||h||_{L^2(\Sp^2)} \;.
\end{split}
\end{equation}
We now go through the proof of \un{Proposition 4.11} with \eqref{EqComEq} in place of $\mathcal{T}_{[2]} \psi = 0$. Differentiating \eqref{EqComEq} twice in $r$, we then use the same multiplier as in \un{(4.14)} with $\psi$ replaced by $\Pb \psi$ and $q_r = q$. The two additional terms arising on the right hand side of \un{(4.14)} are
\begin{equation}
\label{EqTwoAddT}
\Big(- a^2 [ \sin^2 \theta, \Pb] \, \rd_{v_+}^2 \rd_r^2 \psi + 2isa [ \cos \theta, \Pb ] \, \rd_{v_+} \rd_r^2 \psi \Big) \Big( v_+^{q} \big(-(1 + \lambda \Delta) \rd_r + (1 + \lambda \Delta) \rd_{v_+} \big) \overline{ \rd_r^2 \Pb \psi} \Big)
\end{equation}
Using \eqref{EqCommutator} and Cauchy Schwarz,  after integration over the spheres we obtain for $\delta >0$
\begin{equation}\label{EqAI}
\Big| \textnormal{\eqref{EqTwoAddT}} \Big| \underset{a.i.}{\lesssim} v_+^{q} \delta \big( | \rd_r^3 \Pb \psi|^2 + | \rd_{v_+} \rd_r^2 \Pb \psi|^2 \big) + v_+^{q} \frac{1}{\delta} \big( a^4 \cdot | \rd_r^2 \rd_{v_+}^2 \psi |^2 + s^2 a^2 | \rd_r^2 \rd_{v_+} \psi|^2 \big)
\end{equation}
The right hand side of \eqref{EqAI} will appear with a minus sign on the right hand side of the equation above \un{4.18} with $\psi$ replaced by $\Pb \psi$ and $q_r = q$. For $\delta >0$ sufficiently small the first summand in \eqref{EqAI} can be absorbed by the right hand side. The second summand in \eqref{EqAI} is bounded by  \eqref{EqV4.19}   so that we obtain \un{(4.18)} and \un{(4.19)} with an additional constant on the right hand side and with $\psi$ replaced by $\Pb \psi$ and $q_r = q$. The rest follows as in \underline{Step 6}. Thus, \un{Proposition 4.11} holds indeed with $\psi$ replaced by $\Pb \psi$ and $q_r = q$.

We move on to \un{Corollary 4.21}, which combines \un{Proposition 4.1} with \un{Proposition 4.11}. We do exactly the same change of coordinates as in the proof of \un{Corollary 4.21} to translate \un{Proposition 4.1}, with $q_l \geq q$, into $(v_+, r, \theta, \varphi_+)$ coordinates and $\hat{\psi}$ into $\psi$. This gives \un{(4.22)} in the domain $\{v_+ \leq 1\}$. We then 
use
\begin{equation*}
\begin{split}
||\widetilde{Z}_{i,+} \Pb \psi||_{L^2(\Sp^2)} & \leq || \widetilde{Z}_{i,+} \psi||_{L^2(\Sp^2)} + \sum_{m=-2}^2 || \widetilde{Z}_{i,+} \mathbb{P}_{S(m2)}\psi||_{L^2(\Sp^2)} \\
&\leq ||\widetilde{Z}_{i,+}  \psi||_{L^2(\Sp^2)} + \sum_{m=-2}^2 || \psi||_{L^2(\Sp^2)} \cdot || \widetilde{Z}_{i,+}  Y^{[2]}_{m2}||_{L^2(\Sp^2)} \\
&\leq ||\widetilde{Z}_{i,+}  \psi||_{L^2(\Sp^2)} + C_{l=2} ||\psi||_{L^2(\Sp^2)}
\end{split}
\end{equation*}
to obtain
\begin{equation*}
\begin{split}
\sum_{0 \leq i_1 + i_2 + i_3 + j + k \leq 1} \int_{\{r=r'\} \cap \{v_+ \leq 1\}} &\chi(v_+) |\widetilde Z_{1,+}^{i_1}\widetilde Z_{2,+}^{i_2}\widetilde Z_{3,+}^{i_3}\rd_{v_+}^{j}\big(\rd_r|_+\big)^{k}\Pb f|^2 \,\vols\,dv_{+} \\
&\lesssim \sum_{0 \leq i_1 + i_2 + i_3 + j + k \leq 1} \int_{\{r=r'\} \cap \{v_+ \leq 1\}} \chi(v_+) |\widetilde Z_{1,+}^{i_1}\widetilde Z_{2,+}^{i_2}\widetilde Z_{3,+}^{i_3}\rd_{v_+}^{j}\big(\rd_r|_+\big)^{k} f|^2 \,\vols\,dv_{+} \;.
\end{split}
\end{equation*}
Combining this with our previous result of \un{Proposition 4.11} for $\Pb \psi$ we obtain \un{Corollary 4.21} for $ \Pb \psi$ instead of $\psi$ and with $q_r = q$, $q_l \geq q$.

We continue to show that \un{Proposition 4.31} holds with $\Pb \psi$ instead of $\psi$ and $q_r = q$, $q_l \geq q$. The proof is carried out in $(t,r,\theta, \varphi)$ coordinates. Note that $\rd_{v_+} = \rd_t$. Using the same multiplier as in \un{(4.34)} with $\Pb \psi$ instead of $\psi$ and the new $q_r, q_l$ we again obtain from \eqref{EqComEq} that we get \un{(4.34)} with 
$$ \Big(-a^2 [\sin^2 \theta, \Pb] \, \rd_t^2 \psi + 2sia [\cos \theta, \Pb] \, \rd_t \psi\Big)\big( - \chi(t) e^{\lambda r} \rd_r \overline{\Pb \psi}\big)$$ 
on the right hand side. This additional term is estimated in exactly the same way as before, using \eqref{EqV4.32}. The modification in the proof of \un{Proposition 4.51} is analogous, using \eqref{EqV4.52}. Thus we obtain \un{Proposition 4.51} with $\psi$ replaced by $\Pb \psi$ and $q_r = q$, $q_l \geq q$. The proof of the extension to the left Cauchy horizon remains unchanged. This shows \eqref{EqThmConD'}.

We now prove \eqref{EqThmConA'}. To begin with, let us recall that the assumptions made in Theorem \ref{MainThm} are strictly stronger than those in \un{Theorem 3.9} if we set $q_r = q_-$, $p_0 = \lceil \frac{q}{2} \rceil$, and $l_0 = 2$. Thus, all the partial results in \cite{Sbie23} still hold. To improve on the blow-up bound \un{(8.5)}, we begin with an updated version of \un{Proposition 7.5} using our stronger assumptions on the event horizon.

\begin{proposition} \label{PropIPW}
The assumptions \eqref{EqThmA}, \eqref{EqThmB}, \eqref{EqThmC}, together with \underline{Assumption 2.46} imply that there exists an $m_0 \in \{\pm 2, \pm 1, 0\}$ such that  $(\widecheck{\psi|_{\Hp_r}})_{m_02} \notin H^{\f{q}{2}}_\omega \big((-\varepsilon, \varepsilon)\big)$ for any $\varepsilon >0$.
\end{proposition}

\begin{proof}
By assumption \eqref{EqThmA} there exists an $m_0 \in \{\pm 2, \pm 1, 0\}$ such that 
\begin{equation}
\label{EqBlo}
\int_\R |v_+|^{q} | (\psi|_{\Hp_r})_{S(m_0 2)} |^2 \, dv_+ = \infty \;.
\end{equation}
Thus, if $q \in \N_{\geq 4}$ is even, the claim follows directly from \un{Proposition 7.5}. Thus it remains to treat the case $q = 2\mathfrak{p}-1$, $\mathfrak{p} \in \N_{\geq 2}$.

It follows from \eqref{EqThmB} and \eqref{EqThmC}, together with \un{Assumption 2.46} which gives exponential decay of $\psi|_{\Hp_r}$ for $v_+ \to - \infty$, in the same way as in \eqref{EqPropA} and \eqref{EqPropB}, that $(\widecheck{\psi|_{\Hp_r}})_{S(m_02)} \in H^{\mathfrak{p}-1}(\R)$ and $ \omega \cdot \rd_\omega^{\mathfrak{q}} (\widecheck{\psi|_{\Hp_r}})_{S(m_02)} \in H^{\f{1}{2}}(\R)$ for all $0 \leq {\mathfrak{q}} \leq \mathfrak{p}-1$. Combining $(\widecheck{\psi|_{\Hp_r}})_{S(m_02)} \in H^{\mathfrak{p}-1}(\R)$  with \eqref{EqBlo}, Proposition \ref{PropFourierSob} implies $\rd_\omega^{\mathfrak{p}-1} (\widecheck{\psi|_{\Hp_r}})_{S(m_02)} \notin H^{\f{1}{2}}(\R)$. But now we are in the setting of Proposition \ref{PropLocH} with $h = \rd_\omega^{\mathfrak{p}-1} (\widecheck{\psi|_{\Hp_r}})_{S(m_02)} $ and hence $(\widecheck{\psi|_{\Hp_r}})_{S(m_02)}  \notin H^{\mathfrak{p} - \f{1}{2}}\big((-\varepsilon, \varepsilon)\big)$ for any $\varepsilon>0$. Applying Proposition \ref{PropRelateSNew} with $h = \psi|_{\Hp_r}$ then gives $(\widecheck{\psi|_{\Hp_r}})_{(m_02)}  \notin H^{\mathfrak{p} - \f{1}{2}}\big((-\varepsilon, \varepsilon)\big)$ for any $\varepsilon>0$.
\end{proof}
In the case of $q$ even with thus obtain directly \un{(8.4)} with $l_0 =2$ and $p_0 = \frac{q}{2}$. In the case of odd $q = 2 \mathfrak{p} -1$ we claim that we can replace \un{(8.4)} by
\begin{equation}
\label{EqImp84}
\widecheck{\psi}_{m_02}(r_-; \omega) \notin H^{{\mathfrak{p}} - \frac{1}{2}}\big((-\varepsilon, \varepsilon) \big) \quad \textnormal{ for any } \varepsilon >0\;.
\end{equation}
To see this, by \un{(8.1)} we have
\begin{equation}\label{EqNT}
\widecheck{\psi}_{m_02}(r_-; \omega) = \mathfrak{R}_{\Hp_l, m_02}(\omega) a_{\Hp_l, m_02}(\omega) + \mathfrak{T}_{\Hp_r, m_02}(\omega) a_{\Hp_r, m_02}(\omega) \;.
\end{equation}
We first consider $m_0 \neq 0$. Recall that by \un{Proposition 6.22} the reflection and transmission coefficients are smooth on $(-|\omega_+|, |\omega_+|)$. Using also $a_{\Hp_l, m_02} \in H^{\mathfrak{p} - \f{1}{2}}\big((-\varepsilon_0, \varepsilon_0)\big)$, which we get from \un{Corollary 7.29} (recall $q_l \geq q = 2 \mathfrak{p}-1$), in combination with Lemma \ref{LemChiK}, gives that the first product on the right hand side of \eqref{EqNT} is in $H^{\mathfrak{p}-\f{1}{2}}\big((-\varepsilon_0, \varepsilon_0)\big)$ for $|\omega_+|>\varepsilon_0>0$. The second product is however \emph{not} in $H^{{\mathfrak{p}} - \frac{1}{2}}\big((-\varepsilon, \varepsilon) \big) $ for any $\varepsilon>0$: by \un{Proposition 7.4} and Proposition \ref{PropIPW} we have $a_{\Hp_r, m_02}(\omega) = (\widecheck{\psi|_{\Hp_r}})_{(m_02)}(\omega) \notin H^{{\mathfrak{p}}-\f{1}{2}}_\omega \big((-\varepsilon, \varepsilon)\big)$ for any $\varepsilon >0$. Moreover \un{Proposition 6.22} gives $\mathfrak{T}_{\Hp_r, m_02}(0) \neq 0$. The statement then follows from Corollary \ref{CorChi}. This shows \eqref{EqImp84} for $m_0 \neq 0$. The case of $m_0 = 0$ follows similarly (cf.\ the structure in \un{(8.3)} without the $\rd_\omega^{p_0}$).

In summary, we have now shown that $\widecheck{\psi}_{m_02}(r_-; \omega) \notin H^{\frac{q}{2}}\big((-\varepsilon, \varepsilon) \big)$ for any  $\varepsilon >0$. We claim that this implies 
\begin{equation}\label{Eqwww}
\widecheck{\psi}_{S(m_02)}(r_-; \omega) \notin H^{\frac{q}{2}}\big((-\varepsilon, \varepsilon) \big) \quad \textnormal{ for any } \varepsilon >0\;.
\end{equation}
Recall that \un{Proposition 4.64} with $q_r = q_-$ and $q_l \geq q$ holds as well as \eqref{EqV4.64}. In the case $q = 2 \mathfrak{p}-1$ we thus note that $\psi(v_+, \theta, \varphi_+; r_-)$ satisfies the assumptions of Proposition \ref{PropRelateSNew}, which implies \eqref{Eqwww}. In the case of $q$ even, \eqref{Eqwww} follows with the same argument as in \un{Proposition 7.5}, below \un{(7.12)}.

We claim that \eqref{Eqwww} gives  
\begin{equation} \label{Eqrrr}
\int_\R  v_+^{q}  |\psi_{S(m_02)}(v_+, r_-)|^2 \, dv_+ = \infty \;.
\end{equation} 
If this were finite, then in combination with \un{Proposition 4.64} with $q_r = q$ and $q_l \geq q$ we would get $\int_\R (1+ |v_+|^{q}) |\psi_{S(m_02)}(v_+, r_-)|^2 \, dv_+ < \infty$. Then Proposition \ref{PropFourierSob} gives a contradiction to \eqref{Eqwww}.
Moreover, since $q_l \geq q$, \un{Proposition 4.64} together with \eqref{Eqrrr} implies that  
\begin{equation*}
\int_1^\infty  v_+^{q}  |\psi_{S(m_02)}(v_+, r_-)|^2 \, dv_+ = \infty\;,
\end{equation*}
which gives \eqref{EqThmConA'}.

Finally, it remains to propagate the bounds \eqref{EqThmConA'} -- \eqref{EqThmConD'} backwards to $\Sigma$. We compute in $(v_+, r, \theta, \varphi_+)$ coordinates for $f \in \{\psi, \mathbb{P}_{S(l=2)} \psi, \Pb \psi, \rd_{v_+} \psi\}$ 
$$|f(v_+, r', \theta, \varphi_+) - f(v_+, r_\Sigma(v_+, \theta, \varphi_+), \theta, \varphi_+)| \leq \int_{r'}^{r_\Sigma(v_+, \theta, \varphi_+)} |\rd_r f(v_+, \theta, \varphi_+)| \, dr$$
and thus
\begin{equation}
\label{EqLogDafermos}
|f(v_+, r', \theta, \varphi_+) - f(v_+, r_\Sigma(v_+, \theta, \varphi_+), \theta, \varphi_+)|^2 \leq \underbrace{\Big(\sup\limits_{(\theta, \varphi_+) \in \Sp^2} |r_\Sigma(v_+, \theta, \varphi_+) - r'| \Big)}_{\leq C v_+^{-\sigma}} \cdot \int_{r'}^{r_\Sigma(v_+, \theta, \varphi_+)} |\rd_r f(v_+, \theta, \varphi_+)|^2 \, dr \;.
\end{equation}
Further integration then yields
\begin{equation} \label{EqFinal}
\begin{split}
\int_{v'}^\infty \int_{\Sp^2} v_+^{q} &|f(v_+, r', \theta, \varphi_+) - f(v_+, r_\Sigma(v_+, \theta, \varphi_+), \theta, \varphi_+)|^2 \vols dv_+ \\
&\leq \int_{v'}^\infty \int_{\Sp^2} \int_{r'}^{r_\Sigma(v_+, \theta, \varphi_+)} v_+^{q-\sigma}|\rd_r f(v_+, \theta, \varphi_+)|^2 \, dr \vols dv_+ \;.
\end{split}
\end{equation}
For $v'$ large enough and $v' \leq v_+$ we have $r_- < r_{\Sigma}(v_+, \theta, \varphi_+) \leq r_{\mathrm{ered}}$, so that the integral on the right hand side is bounded by \un{(4.52)} in the case of $f \in \{\psi, \mathbb{P}_{S(l=2)} \psi, \Pb \psi\}$ and by \eqref{EqV4.52} in the case of $f = \rd_{v_+} \psi$. Now, \eqref{EqThmConA} -- \eqref{EqThmConD} follow from \eqref{EqThmConA'} -- \eqref{EqThmConD'} using \eqref{EqFinal}.
This concludes the proof of Theorem \ref{MainThm}.

\section*{Acknowledgments}

The author acknowledges the support through the Royal Society University Research Fellowship URF\textbackslash R1\textbackslash 211216.

\bibliographystyle{acm}
\bibliography{Bibly}

\end{document}